\documentclass[11pt]{article}
\usepackage{amsmath,amsfonts,amsthm,a4wide,amssymb}
\usepackage{enumitem,xspace,url,dsfont,engord}

\setenumerate{topsep=3pt,parsep=-2pt,leftmargin=*}
\setenumerate[3]{label=(\roman*)}
\setenumerate[2]{label=(\alph*)}
\setenumerate[1]{label=(\arabic*)}

\newcommand{\set}[1]{\left\{#1\right\}}

\renewcommand{\P}{\mbox{\sf P}\xspace}
\newcommand{\NP}{\mbox{\sf NP}\xspace}
\newcommand{\nP}{\mbox{\sf \#P}\xspace}
\newcommand{\nPQ}{\mbox{$\nP_\Rationals$}\xspace}
\newcommand{\CSP}{\mbox{\sf CSP}\xspace}
\newcommand{\nCSP}{\mbox{\sf \#CSP}\xspace}
\newcommand{\WCSP}{\textsf{weighted}\,\nCSP}
\newcommand{\FP}{\mbox{\sf FP}\xspace}
\newcommand{\FPnP}{\ensuremath{\FP^{\mbox{\scriptsize\sf \#P}}}\xspace}

\def\redT{\leq_\mathrm{T}}
\def\Rationals{\ensuremath{\mathbb{Q}}\xspace}

\def\xor{\oplus}
\def\substuple#1#2#3{{#1}^{#2=#3}}
\newcommand{\eval}{{\sc Eval}}

\newtheorem{theorem}{Theorem}
\newtheorem{lemma}[theorem]{Lemma}
\newtheorem{corollary}[theorem]{Corollary}

\newtheorem*{appendixlemma}{Lemma~\ref{lem:pinning}}

\title{The Complexity of Weighted Boolean \#CSP}
\author{Martin Dyer \\
School of Computing\\ University of Leeds\\
Leeds LS2~9JT, UK\and
Leslie Ann Goldberg \\
Department of Computer Science,\\
 University of Liverpool,\\
 Liverpool
L69 3BX, UK\and
Mark Jerrum \\
School of Mathematical Sciences,\\
Queen Mary, University of London\\
Mile End Road, London E1 4NS, UK}
\date{28 April 2008}
\begin{document}

\maketitle

\begin{abstract}
This paper gives a dichotomy theorem for the complexity of
computing
the partition function of an instance of a weighted Boolean constraint
satisfaction problem.  The problem is parameterised by a finite set~$\mathcal{F}$
of non-negative functions that may be used to assign
weights to the configurations (feasible solutions)
of a problem instance.
Classical constraint satisfaction problems correspond to the
special case of 0,1-valued functions.
We show that computing the partition function, i.e. the sum of the
weights of all configurations,
is \FPnP-complete unless
either
(1)~every function in~$\mathcal{F}$ is of ``product type'', or
(2) every function in~$\mathcal{F}$ is ``pure affine''.
In the remaining cases,
computing the partition function is
in \P.
\end{abstract}

\section{Introduction}
This paper gives a dichotomy theorem for the complexity of the
partition function of weighted Boolean constraint satisfaction
problems. Such problems are parameterised by a set~$\mathcal{F}$ of
non-negative functions that may be used to assign weights to
configurations (solutions) of the instance. These functions take the
place of the allowed constraint relations in classical constraint
satisfaction problems (CSPs). Indeed, the classical setting may be
recovered by restricting $\mathcal{F}$ to functions with range
$\{0,1\}$. The key problem associated with an instance of a weighted
CSP is to compute its partition function, i.e., the sum of weights
of all its configurations.
Computing the partition function
of a weighted CSP
may be viewed a generalisation of counting the number of satisfying solutions of a
classical
CSP.
Many partition functions from
statistical physics may be expressed as weighted CSPs. For example,
the \emph{Potts model}~\cite{welsh} is naturally expressible as a weighted CSP, whereas
in the classical framework only the ``hard core'' versions may be
directly expressed.
(The hard-core version of the
\emph{antiferromagnetic} Potts model corresponds to graph colouring
and the hard-core version of the \emph{ferromagnetic} Potts model is
trivial --- acceptable configurations colour the entire graph with a
single colour.)
A corresponding weighted version of the decision CSP was investigated by Cohen, Cooper, Jeavons and Krokhin~\cite{CoCoJK06}. This results in optimisation problems.

We use $\nCSP(\mathcal{F})$ to denote the problem of
computing the partition function of weighted CSP instances that
can be expressed using only functions from~$\mathcal{F}$.
We show in
Theorem~\ref{thm:main}
below
that if every function $f\in
\mathcal{F}$ is ``of product type'' then computing the partition
function $Z(I)$ of an instance~$I$ can be done in polynomial time.
Formal definitions are given later, but the condition of being ``of
product type'' is easily checked --- it essentially means that the
partition function factors.
We show further in
Theorem~\ref{thm:main}
that
if every function $f\in\mathcal F$ is ``pure affine'' then the
partition function of $Z(I)$ can be computed in polynomial
time. Once again, there is an algorithm to check whether $\mathcal
F$ is pure affine. For each other set~$\mathcal F$,
we show in Theorem~\ref{thm:main} that
computing the
partition function of a $\nCSP(\mathcal F)$ instance is complete for the class \FPnP.
The existence of algorithms for testing the properties of being
purely affine or of product type means that the dichotomy
is effectively decidable.

\subsection{Constraint satisfaction}

\emph{Constraint Satisfaction}, which originated in Artificial
Intelligence, provides a general framework for modelling decision
problems, and has many practical applications. (See, for
example~\cite{RoBeWa06}.) Decisions are modelled by
\emph{variables}, which are subject to \emph{constraints}, modelling
logical and resource restrictions. The paradigm is sufficiently
broad that many interesting problems can be modelled, from
satisfiability problems to scheduling problems and graph-theory
problems. Understanding the complexity of constraint satisfaction
problems has become a major and active area within computational
complexity~\cite{cks, hnbook}.

A Constraint Satisfaction Problem (CSP)
typically has a finite \emph{domain}, which we will
denote by $[q]=\set{0,1\ldots,q-1}$ for a positive integer~$q$.\footnote{
Usually $[q]$ is defined to be $\set{1,2,\ldots,q}$, but it is more convenient here
to start the enumeration of domain elements at~0 rather than~1.}
A \emph{constraint
language\/} $\Gamma$
with domain $[q]$
is a set of relations on~$[q]$. For example,
take $q=2$.
The
relation $R=\{(0,0,1)$, $(0,1,0)$, $(1,0,0)$, $(1,1,1)\}$ is a 3-ary
relation on the domain $\{0,1\}$, with four tuples.

Once we have fixed a constraint language $\Gamma$,
an \emph{instance\/} of the CSP is a set of \emph{variables\/}
$V=\{v_1,\ldots,v_n\}$ and a set of \emph{constraints}. Each
constraint has a \emph{scope,} which is a tuple of variables (for
example, $(v_4, v_5, v_1)$) and a relation from~$\Gamma$ of the same
arity, which constrains the variables in the scope. A
\emph{configuration} $\sigma$ is a function from~$V$ to~$[q]$. The
configuration~$\sigma$ is \emph{satisfying} if the scope of every
constraint is mapped to a tuple that is in the corresponding
relation.
In our example above, a configuration $\sigma$ satisfies the constraint with scope
$(v_4,v_5,v_1)$ and relation~$R$
if and only if it maps an odd number of the variables in
$\{v_1,v_4,v_5\}$ to the value~$1$.
Given an instance of a CSP with
constraint language $\Gamma$, the \emph{decision problem}
$\CSP(\Gamma$) asks us to determine whether any configuration is
satisfying. The \emph{counting problem} $\nCSP(\Gamma$) asks us to
determine the \emph{number} of (distinct) satisfying configurations.

Varying the constraint language~$\Gamma$ defines the classes \CSP and
\nCSP of decision and counting problems. These contain problems of
different computational complexities. For example, if
$\Gamma=\{R_1,R_2,R_3\}$ where $R_1$, $R_2$ and $R_3$ are the three
binary relations defined by $R_1=\{(0,1),(1,0),(1,1)\}$,
$R_2=\{(0,0),(0,1),(1,1)\}$ and $R_3=\{(0,0),(0,1),(1,0)\}$, then
$\CSP(\Gamma)$ is the classical 2-Satisfiability problem, which is in
\P. On the other hand, there is a similar constraint
language~$\Gamma'$ with four relations of arity~3 such that
3-Satisfiability (which is \NP-complete) can be represented in
$\CSP(\Gamma')$.  It may happen
that the counting problem is harder than the decision problem.
If $\Gamma$ is the constraint language of 2-Satisfiability above, then $\nCSP(\Gamma)$
contains the problem of counting independent sets in graph,
and is
\nP-complete~\cite{Valian79},
even if restricted to 3-regular graphs~\cite{Greenh00}.

Any  decision problem $\CSP(\Gamma)$ is in \NP, but not every problem
in \NP can be represented as a CSP. For example, the question ``Is
$G$ Hamiltonian?'' cannot
naturally
be expressed as a CSP, because the
property of being Hamiltonian cannot be captured by relations of
bounded size. This limitation of the class \CSP has an important
advantage. If $\P \neq \NP$, then there are problems
which are neither in \P nor \NP-complete~\cite{L75}. But, for
well-behaved smaller classes of decision problems, the situation can
be simpler. We may have a \emph{dichotomy theorem}, partitioning all
problems in the class into those which are in \P and those which are
\NP-complete. There are no ``leftover'' problems of intermediate
complexity. It has been conjectured that there is a dichotomy
theorem for \CSP.  The conjecture is  that $\CSP(\Gamma)$ is in \P
for some constraint languages $\Gamma$, and $\CSP(\Gamma)$ is \NP-complete
for all other constraint languages $\Gamma$. This conjecture appeared in a
seminal paper of Feder and Vardi~\cite{fv}, but has not yet been proved.

A similar
dichotomy, between \FP and \nP-complete, is conjectured for
\#CSP~\cite{BD}.  The complexity classes \FP and \nP are the
analogues of \P and \NP for counting problems.
\FP is simply the class of functions computable in deterministic
polynomial time. \nP is the class of integer functions that can be
expressed as the number of accepting computations of
a polynomial-time non-deterministic Turing machine.
Completeness in \nP is defined with respect to polynomial-time
Turing reducibility~\cite[Chap.~18]{Pa94}. 
Bulatov and Dalmau~\cite{BD} have shown in one
direction that, if $\nCSP(\Gamma)$ is solvable in polynomial time,
then the constraints in $\Gamma$ must have certain algebraic
properties (assuming $\P\neq\nP$). In particular, they must have a so-called \emph{Mal'tsev
polymorphism}. The converse is known to be false,
though it remains
possible that the dichotomy (if it exists) does have an algebraic
characterisation.

The conjectured dichotomies for \CSP and \nCSP are major open
problems for computational complexity theory. There have been many
important results for subclasses of \CSP and \nCSP. We mention the
most relevant to our paper here. The first decision dichotomy was that
of Schaefer~\cite{schaefer}, for the Boolean domain $\{0,1\}$.
Schaefer's result is as follows.
\begin{theorem}[Schaefer~\cite{schaefer}]
\label{thm:schaefer} Let $\Gamma$ be a constraint language with domain
$\{0,1\}$. The problem $\CSP(\Gamma)$ is
in \P if $\Gamma$ satisfies one of the conditions below. Otherwise,
$\CSP(\Gamma)$ is \NP-complete.
\begin{enumerate}[topsep=5pt]
\item $\Gamma$ is $0$-valid or $1$-valid.
\item $\Gamma$ is weakly positive or weakly negative.
\item $\Gamma$ is affine.
\item $\Gamma$ is bijunctive.
\end{enumerate}
\end{theorem}

We will not give detailed definitions of the conditions in
Theorem~\ref{thm:schaefer}, but the interested reader is referred to
the paper~\cite{schaefer} or to Theorem~6.2 of the
textbook~\cite{cks}. An interesting feature is that the conditions
in~\cite[Theorem~6.2]{cks} are all checkable. That is, there is an algorithm to
determine whether $\CSP(\Gamma$) is in \P or \NP-complete, given a constraint
language~$\Gamma$ with domain~$\{0,1\}$.
Creignou and Hermann~\cite{CH} adapted Schaefer's decision
dichotomy to obtain a counting dichotomy for the Boolean domain.
Their result is as follows.
\begin{theorem}[Creignou and Hermann~\cite{CH}]
\label{thm:CH}
Let $\Gamma$ be a constraint language with domain
$\{0,1\}$.
The problem $\nCSP(\Gamma)$ is in\/ \FP\/
if\/ $\Gamma$ is affine. Otherwise, $\nCSP(\Gamma)$ is \nP-complete.
\end{theorem}

A constraint language~$\Gamma$ with domain $\{0,1\}$ is \emph{affine} if
every relation $R\in \Gamma$ is affine.
A relation $R$ is affine if the set of
tuples~$x\in R$ is the set of solutions to a system of linear
equations over GF($2$). These equations are of the form $v_1 \oplus
\cdots \oplus v_n =0$ and $v_1 \oplus \cdots \oplus v_n =1$ where
$\oplus$ is the \emph{exclusive or} operator. It is well known (see,
for example, Lemma~4.10 of~\cite{cks}) that a relation $R$ is affine
iff $a,b,c\in R$ implies $d=a\oplus b\oplus c\in R$. (We will use
this characterisation below.) There is an algorithm for determining
whether a Boolean constraint language $\Gamma$ is affine, so there is an
algorithm for determining whether $\nCSP(\Gamma)$ is in \FP or
\nP-complete.

\subsection{Weighted \nCSP}

The weighted framework of \cite{BG05} extends naturally to
Constraint Satisfaction Problems. Fix the domain $[q]$.
Instead of constraining a length-$k$ scope with an arity-$k$
relation on~$[q]$, we give a weight to the configuration on this scope
by applying a function $f$ from $[q]^{k}$ to the non-negative rationals.
Let $\mathcal{F}_q = \{ f: [q]^k \rightarrow
\Rationals^+ \mid k\in \mathbb{N}\}$
be the set of all such functions (of all arities).\footnote{%
We assume $0\in\mathbb{N}$, so we allow non-negative constants.}
Given a function~$f\in\mathcal{F}_q$ of arity~$k$, the
\emph{underlying relation} of~$f$ is given by $R_{f}=\{x\in[q]^k
\mid f(x)\not=0\}$. It is often helpful to think of $R_f$ as a
table, with $k$ columns corresponding to the positions of a
$k$-tuple. Each row corresponds to a tuple $x=(x_1,\ldots,x_k)\in
R_f$. The entry in row $x$ and column $j$ is $x_j$, which is a value
in $[q]$.

A \emph{weighted} \#CSP problem is parameterised by a finite subset
$\mathcal{F}$ of $\mathcal{F}_q$, and will be denoted by $\nCSP(\mathcal{F})$.
An instance $I$ of $\nCSP(\mathcal{F})$ consists of a set~$V$
of \emph{variables} and a set~$\mathcal{C}$ of \emph{constraints}.
Each constraint $C\in \mathcal{C}$ consists of a function
$f_C\in \mathcal{F}$ (say of arity~$k_C$) and a \emph{scope}, which is a sequence
$s_C=(v_{C,1},\ldots,v_{C,k_C})$ of variables from~$V$. The variables
$v_{C,1},\ldots,v_{C,k_C}$ need not be distinct. As in the
unweighted case,  a \emph{configuration} $\sigma$ for the
instance~$I$ is a function from $V$ to $[q]$.
The \emph{weight} of
the configuration $\sigma$ is given by
\[w(\sigma)=\prod_{C\in
\mathcal{C}} f_C(\sigma(v_{C,1}),\ldots,\sigma(v_{C,k_C})).\]
Finally, the \emph{partition function} $Z(I)$ is given, for instance $I$,  by
\begin{equation}
\label{CSPZ} Z(I)=\sum_{\sigma:V\rightarrow [q]} w(\sigma).
\end{equation}
In the computational problem $\nCSP(\mathcal{F})$, the
goal is to compute $Z(I)$, given an instance~$I$.

Note that an (unweighted) CSP counting problem $\nCSP(\Gamma)$ can
be represented naturally as a weighted CSP counting problem. For
each relation~$R\in \Gamma$, let $f^R$ be the indicator function for
membership in~$R$. That is, if $x\in R$ we set $f^R(x)=1$. Otherwise
we set $f^R(x)=0$. Let $\mathcal{F}=\{f^R \mid R \in \Gamma\}$. Then
for any instance $I$ of $\nCSP(\Gamma)$, the number of satisfying
configurations for~$I$ is given by the (weighted) partition
function $Z(I)$ from~(\ref{CSPZ}).

This framework has been employed previously in connection with \emph{graph homomorphisms}~\cite{bw}. Suppose $H=(H_{ij})$ is any symmetric square matrix~$H$ of rational numbers. We view $H$ as being an edge-weighting of an undirected graph $\mathcal{H}$, where a zero weight in $H$ means that the corresponding edge is absent from $\mathcal{H}$.  Given a (simple) graph $G=(V,E)$ we consider computing the partition function
\begin{equation*}
Z_H(G) = \sum_{\sigma:V\rightarrow[q]}w(\sigma),\quad\textrm{where}\ \
w(\sigma)=\prod_{\{u,v\}\in E} H_{\sigma(u)\sigma(v)}.
\end{equation*}
Within our framework above, we view $H$ as the binary function $h:[q]^2\to \mathbb{R}$, and the problem is then computing the partition function of $\#CSP(\set{h})$.

Bulatov and Grohe~\cite{BG05} call~$H$
\emph{connected} if $\mathcal{H}$ is connected and
\emph{bipartite} if $\mathcal{H}$ is bipartite. They give
the following dichotomy theorem for non-negative~$H$.\footnote{This is not quite the original statement of the theorem. We have chosen here to restrict all inputs to be rational, in order to avoid issues of how to represent, and compute with, arbitrary real numbers.}
\begin{theorem}[Bulatov and Grohe~\cite{BG05}]
\label{thm:bulgro} Let $H$ be a symmetric matrix with non-negative
rational entries.
\begin{enumerate}
\item If $H$ is connected and not bipartite, then computing $Z_H$ is in\/
\FP if the rank of~$H$ is at most~$1$; otherwise computing $Z_H$ is \nP-hard.
\item If $H$ is connected and bipartite, then computing $Z_H$
is in \FP if
the rank of~$H$ is at most~$2$; otherwise computing $Z_H$ is \nP-hard.
\item If $H$ is not connected, then computing $Z_H$ is in \FP
if each of its connected components satisfies the corresponding
conditions stated in (1) or~(2); otherwise computing $Z_H$ is \nP-hard.\qed
\end{enumerate}
\end{theorem}

Many partition functions arising in statistical physics may be viewed as
weighted \nCSP problems. An example is the $q$-state Potts model
(which is, in fact, a weighted graph homomorphism problem). In
general, weighted \nCSP is very closely related to the problem of
computing the partition function of a Gibbs measure in the framework
of Dobrushin, Lanford and Ruelle (see~\cite{bw}).
See also the framework of Scott and Sorkin~\cite{ss}.

\subsection{Some Notation}

We will call the class of (rational) weighted \#CSP problems
\emph{weighted} \nCSP. The sub-class having domain size $q=2$ will be called weighted \emph{Boolean} \nCSP, and will be the main focus of this paper.
We will give a dichotomy theorem for weighted Boolean \nCSP.

Since weights can be arbitrary non-negative rational
numbers, the solution to these problems is not an integer in general.
Therefore $\nCSP(\mathcal{F})$ is not necessarily in the class~\nP.
However, Goldberg and Jerrum~\cite{GJ06} have observed that $Z(I)=\tilde{Z}(I)/K(I)$, where
$\tilde{Z}$ is a function in \nP and $K(I)$ is a positive integer computable in \FP.
This follows because, for all $f\in\mathcal{F}$, we can ensure that
$f(\cdot)=\tilde{f}(\cdot)/K(I)$, where $\tilde{f}(\cdot)\in\mathbb{N}$,
by``clearing denominators''. The denominator $K(I)$
can obviously be computed in polynomial time, and it is straightforward to show
that computing $\tilde{Z}(I)$ is in $\nP$, so the characterisation of~\cite{GJ06} follows.
The resulting complexity class, comprising functions which are a function in \nP divided by a function in \FP, is named \nPQ in~\cite{GJ06}, where it is used in the context of approximate counting. Clearly we have
\[ \WCSP\ \subseteq\ \nPQ\ \subseteq\ \FPnP.\]
On the other hand, if $Z(I)\in\WCSP$ is \nP-hard, then, using an
oracle for computing $Z(I$), we can construct a \nP oracle
$\tilde{Z}(I)$ as outlined above.
(Note that $Z(I)\notin \nP$ in general.)
Using this, we can compute any
function in \FPnP with a polynomial time-bounded oracle Turing
machine. Thus any \nP-hard function in \textsf{weighted}\,\nCSP is
complete for \FPnP. We will use this observation to state our main
result in terms of completeness for the class \FPnP.

We make the following definition, which relates to the discussion above. We will say that $\mathcal{F}\subseteq \mathcal{F}_q$ \emph{simulates} $f\in\mathcal{F}_q$ if, for each instance $I$ of $\nCSP(\mathcal{F}\cup\set{f})$, there is
a polynomial time computable
instance $I'$ of $\nCSP(\mathcal{F})$,
such that $Z(I)=\varphi(I)Z(I')$ for some $\varphi(I)\in\Rationals$
which is \FP-computable. This generalises the notion of \emph{parsimonious reduction}~\cite{Pa94} among problems in \nP. We will use $\redT$ to denote the relation ``is polynomial-time Turing-reducible to'' between computational problems. Clearly, if $\mathcal{F}$ simulates $f$, we have $\nCSP(\mathcal{F}\cup\set{f})\redT\nCSP(\mathcal{F})$.
Note also that, if $\tilde{f}=Kf$, for some constant $K>0$, then
$\set{f}$ simulates $\tilde{f}$. Thus there is no need to distinguish between
``proportional'' functions.

We use the following terminology for certain functions.
Let $\chi_{=}$ be the binary \emph{equality} function defined on $[q]$ as follows. For any element
$c\in[q]$, $\chi_=(c,c)=1$ and for any pair $(c,d)$ of distinct
elements  of~$[q]$, $\chi_{=}(c,d)=0$. Let $\chi_{\not=}$ be the binary
\emph{disequality} function given by $\chi_{\not=}(c,d)=1-\chi_{=}(c,d)$
for all $c,d\in[q]$.\footnote{%
A more general disequality function is defined in the Appendix.}
We say that a function~$f$ is of {\it product type\/} if $f$ can be
expressed as a product of unary functions and binary functions of
the form~$\chi_=$ and $\chi_{\neq}$.

We focus attention in this paper on the Boolean case, $q=2$. In this case,
we say that a function~$f\in\mathcal{F}_2$ has \emph{affine support}
if its underlying relation
$R_f$, defined earlier, is affine.
We say that $f$ is \emph{pure
affine} if it has affine support and range $\{0,w\}$
for some~$w>0$.
Thus a function is pure affine if and only if it is a
positive real multiple of some (0,1-valued) function
which is affine over GF(2).

\subsection{Our Result}

Our main result is the following.
\begin{theorem}
\label{thm:main} Suppose $\mathcal{F}\subseteq \mathcal{F}_2
= \{f : \set{0,1}^k\to \mathbb{Q}^+\mid k \in \mathbb{N}\}$.
If
every function in $\mathcal{F}$ is of product type then
$\nCSP(\mathcal{F})$ is in \FP. If every function in
$\mathcal{F}$ is pure affine then $\nCSP(\mathcal{F})$ is in
\FP. Otherwise, $\nCSP(\mathcal{F})$ is \FPnP-complete.
\end{theorem}

\begin{proof}
Suppose first that $\mathcal{F}$ is of product type.
In this case the partition function $Z(I)$ of an instance~$I$ with variable
set~$V$
is easy to evaluate
because it can be factored into easy-to-evaluate pieces:
Partition the variables in~$V$ into equivalence classes
according to whether or not they are related by an equality or disequality
function.
(The equivalence relation on variables here is ``depends linearly
 on''.)  An equivalence class consists of
two (possibly empty) sets of variables $U_1$ and $U_2$. All of the variables
in~$U_1$ must be assigned the same value by a configuration~$\sigma$
of nonzero weight,
and
all
variables in~$U_2$ must be assigned the other value.
Variables in $U_1\cup U_2$ are not related by equality or disequality
to variables in $V\setminus(U_1\cup U_2)$.
The equivalence class contributes one weight, say $\alpha$, to the partition function
if variables in $U_1$ are given value ``$0$'' by~$\sigma$ and it contributes
another weight, say $\beta$, to the partition function
if variables in
$U_1$
are given value ``$1$'' by~$\sigma$.
Thus, $Z(I)=(\alpha+\beta)Z(I')$, where $I'$ is the instance formed from~$I$ by
removing this equivalence class.
Therefore, suppose we choose any equivalence class and remove its variables.
Since $\mathcal{F}$ contains only unary, equality or binary disequality constraints, we can also remove all functions involving variables in $U_1\cup U_2$ to give $\mathcal{F}'$. Then $I'$ is of product type with fewer variables, so we may compute $Z(I')$ recursively.

Suppose second that $\mathcal{F}$ if pure affine.
Then $Z(I)=\prod_{f\in\mathcal{F}} w_f^{k_f} Z(I')$, where
$\{0,w_f\}$ is the range of~$f$, $k_f$ is
the number of constraints involving $f$ in~$I$, and $I'$ is the instance
obtained from~$I$ by replacing every function~$f$ by its underlying
relation~$R_f$ (viewed as a function with range $\{0,1\}$).
$Z(I')$ is easy to evaluate, because this is just counting
solutions to a linear system
over GF($2$), as Creignou and Hermann have observed~\cite{CH}.

Finally, the  \nP-hardness in Theorem~\ref{thm:main} follows from
Lemma~\ref{seven} below.
\end{proof}

\begin{lemma}
\label{seven}
If $f\in\mathcal{F}_2$ is not of product
type and $g\in\mathcal{F}_2$ is not pure affine then $\nCSP(\{f,g\})$ is
\nP-hard.
\end{lemma}

Note that the functions $f$ and $g$ in Lemma~\ref{seven} may be one
and the same function. So $\nCSP(\{f\})$ is \nP-hard when $f$ is not
of product type nor pure affine. The rest of this article
gives the proof of Lemma~\ref{seven}.

\section{Useful tools for proving hardness of \nCSP}

\subsection{Notation}

For any sequence $u_1,\ldots,u_k$ of variables of~$I$ and
any  sequence $c_1,\ldots,c_k$ of elements of the domain~$[q]$,
we will let
$Z(I \mid \sigma(u_1)=c_1,\ldots,\sigma(u_k)=c_k)$ denote the contribution to~$Z(I)$
from assignments~$\sigma$
with $\sigma(u_1)=c_1,\cdots,\sigma(u_k)=c_k$.

\subsection{Projection}
\label{sec:project}

The first tool that we study is projection,
which is referred
to as ``integrating out'' in the statistical physics literature.

Let $f$ be a function of arity~$k$, and let
$J=\{j_1,\ldots,j_r\}$ be a size-$r$ subset of $\{1,\ldots,k\}$,
where $j_1<\cdots<j_r$.\footnote{It is not necessary to choose this particular ordering for $J$, 
but it is convenient to do so.}
We say that a $k$-tuple $x'\in [q]^k$ {\it extends\/} an $r$-tuple $x\in[q]^r$ on~$J$
(written $x'\sqsupseteq_J x$)
if $x'$ agrees with~$x$ on indices in~$J$;
that is to say, $x'_{j_i}=x_i$ for all
$1\leq i \leq r$.
The {\it projection\/} $g$ of $f$ onto $J$
is defined as follows. For every $x\in[q]^r$,
$g(x) = \sum_{x'\sqsupseteq_J x} f(x')$.

The following lemma
may be viewed as a weighted version of
Proposition~2 of~\cite{BD}, where it is proved for the unweighted case. It is expressed somewhat differently in~\cite{BD}, in terms of counting
the number of solutions to an existential formula.
\begin{lemma}
\label{lem:project}
Suppose $\mathcal{F}\subseteq\mathcal{F}_q$.
Let $g$ be a projection
of a function $f\in\mathcal{F}$ onto a subset of its indices.
Then $\nCSP(\mathcal{F}\cup\{g\}) \redT \nCSP(\mathcal{F})$.
\end{lemma}

\begin{proof}
Let $k$ be the arity of $f$ and let $g$ be the projection of~$f$ onto
the subset~$J$ of its indices.
Let $I$ be an instance of $\nCSP(\mathcal{F}\cup\{g\})$.
We will construct an instance $I'$ of $\nCSP(\mathcal{F})$ such that
$Z(I)=Z(I')$.
The instance
$I'$ is identical to~$I$ except that every constraint~$C$ of~$I$
involving~$g$ is replaced with a new constraint~$C'$ of~$I'$ involving~$f$.
The corresponding scope $(v_{C',1},\ldots,v_{C',k})$
is constructed as follows. If
$j_{\ell}$ is the $\ell$'th element of~$J$,
then $v'_{C',j_{\ell}} = v_{C,\ell}$. The other variables, $v_{C',j}$ ($j\notin J$), are
distinct  new variables. We have shown that $\mathcal{F}$ simulates $g$
with $\phi(I)=1$.
\end{proof}

\subsection{Pinning}

For $c\in[q]$, $\delta_c$ denotes the unary function with
$\delta_c(c)=1$ and $\delta_c(d)=0$ for $d\neq c$.
The following lemma, which allows ``pinning'' CSP variables to specific values in
hardness proofs, generalises Theorem~8 of~\cite{BD},
which does the unweighted case.  Again~\cite{BD} employs different
terminology, and its theorem is a statement about the full idempotent
reduct of a finite algebra. The idea of pinning was used
previously by Bulatov and Grohe~of \cite{BG05} in the context of
counting weighted graph homomorphisms (see Lemma~32 of~\cite{BG05}).
A similar idea was used by Dyer and Greenhill in the
context of counting \emph{unweighted} graph homomorphisms --- in
that context, Theorem~4.1 of~\cite{DG} allows pinning all variables
to a particular \emph{component} of the target graph~$H$.

\begin{lemma}
\label{lem:pinning}
For every $\mathcal{F}\subseteq \mathcal{F}_q$,
$\nCSP(\mathcal{F}\cup \bigcup_{c\in[q]}\delta_c) \redT \nCSP(\mathcal{F})$.
\end{lemma}

The proof of Lemma~\ref{lem:pinning} is deferred to the appendix.
Since we only use the case $q=2$ in this paper, we provide the (simpler)
proof for the Boolean case here.

\begin{lemma}
\label{lem:2pinning}
For every $\mathcal{F}\subseteq \mathcal{F}_2$,
$\nCSP(\mathcal{F}\cup \{\delta_0,\delta_1\}) \redT \nCSP(\mathcal{F})$.
\end{lemma}

\begin{proof}

For
$x\in[2]^k$, let $\overline{x}$ be the $k$-tuple
whose $i$'th component, $\overline{x}_i$, is $x_i \xor 1$, for all~$i$.
Say that $\mathcal{F}$ is \emph{symmetric}
if it is the case that for every arity-$k$ function $f\in\mathcal{F}$
and every $x\in[2]^k$, $f(\overline{x})=f(x)$.

Given an instance $I$ of $\nCSP(\mathcal{F}\cup \{\delta_0,\delta_1\})$
with variable set~$V$
we consider two instances $I'$ and $I''$ of
$\nCSP(\mathcal{F})$.
Let $V_0$ be the set of variables $v$ of~$I$ to which the constraint
$\delta_0(v)$ is applied.
Let $V_1$ be the set of variables $v$ of~$I$ to which the constraint
$\delta_1(v)$ is applied.
We can assume without loss of generality that $V_0$ and $V_1$ do not
intersect. (Otherwise, $Z(I)=0$ and we can determine this without
using an oracle for $\nCSP(\mathcal{F})$.)
Let $V_2=V\setminus(V_0\cup V_1)$.
The instance $I'$ has variables $V_2 \cup \{t_0,t_1\}$
where $t_0$ and $t_1$ are distinct new variables that are not in~$V$.
Every constraint~$C$ of~$I$
involving a function $f\in \mathcal{F}$
corresponds to a constraint~$C'$ of~$I'$.
$C'$ is the same as~$C$ except that variables in~$V_0$ are replaced with~$t_0$
and variables in~$V_1$ are replaced with~$t_1$.
Similarly, the instance $I''$ has variables $V_2 \cup \{t\}$ where $t$ is
a new variable that is not in~$V$.
Every constraint~$C$ of~$I$
involving a function $f\in \mathcal{F}$
corresponds to a constraint~$C''$ of~$I''$.
The constraint~$C''$ is the same
as~$C$ except that variables in~$V_0\cup V_1$ are replaced with~$t$.

\noindent{\bf Case 1.
$\mathcal{F}$ is symmetric:\quad}
By construction,\vspace{-1ex}
\[Z(I')-Z(I'') = Z(I' \mid \sigma(t_0)=0,\sigma(t_1)=1) + Z(I'\mid \sigma(t_0)=1,\sigma(t_1)=0).\vspace{-1ex} \]
By symmetry, the summands are the same, so\vspace{-1ex}
\[Z(I')-Z(I'') = 2 Z(I' \mid \sigma(t_0)=0, \sigma(t_1)=1) =
2 Z(I).\]

\noindent{\bf Case 2.
$\mathcal{F}$ is not symmetric:\quad}
Let $f$ be an arity-$k$ function in~$\mathcal{F}$ and
let $x\in[2]^k$ so that
$f(x)> f(\overline{x})\geq 0$.
Let $s=(t_{x_1},\ldots,t_{x_k})$ and
let $I'_x$ be the instance derived from $I'$ by adding
a new constraint with function~$f$ and
scope~$s$.
Similarly, let $I''_x$ be the instance derived from~$I''$ by
adding a new constraint with function~$f$
and scope $(t,\ldots,t)$. Now
\begin{align*}
Z(I'_x) &=
Z(I' \mid \sigma(t_0)=0, \sigma(t_1)=1)f(x) +
Z(I'\mid \sigma(t_0)=1, \sigma(t_1)=0)f(\overline{x}) \\ &\quad +
Z(I'\mid \sigma(t_0)=0, \sigma(t_1)=0) f(0,\ldots,0) +
Z(I'\mid \sigma(t_0)=1, \sigma(t_1)=1) f(1,\ldots,1)
\\
&=
Z(I' \mid \sigma(t_0)=0, \sigma(t_1)=1)f(x) +
Z(I'\mid \sigma(t_0)=1, \sigma(t_1)=0)f(\overline{x})  +
Z(I''_x).
\end{align*}
Thus we have two independent equations,
\begin{align*}
    Z(I'_x)-Z(I''_x)\ &=\ Z(I' \mid \sigma(t_0)=0, \sigma(t_1)=1)f(x) +
Z(I'\mid \sigma(t_0)=1, \sigma(t_1)=0)f(\overline{x}), \\
    Z(I')-Z(I'')\ &=\ Z(I' \mid \sigma(t_0)=0,\sigma(t_1)=1)\phantom{f(x)}
    + Z(I'\mid \sigma(t_0)=1,\sigma(t_1)=0)\phantom{f(\overline{x})},
\end{align*}
in the unknowns
$Z(I' \mid \sigma(t_0)=0, \sigma(t_1)=1)$ and
$Z(I'\mid \sigma(t_0)=1, \sigma(t_1)=0)$.
Solving these, we obtain the value of
$Z(I' \mid \sigma(t_0)=0, \sigma(t_1)=1) = Z(I)$.
\end{proof}

\subsection{\nP-hard problems}

To prove Lemma~\ref{seven}, we will give reductions from some known
\nP-hard problems. The first of these is the problem of counting
homomorphisms from simple graphs to $2$-vertex multigraphs. We use the following
special case of Bulatov and Grohe's Theorem~\ref{thm:bulgro}.

\begin{corollary}[Bulatov and Grohe~\cite{BG05}]\label{bulgro}
Let $H$ be a symmetric $2\times 2$ matrix with non-negative real
entries.  If $H$ has rank~2 and at most one entry of~$H$ is~$0$
then \eval($H$) is \nP-hard.
\end{corollary}

We will also use the problem of computing the \emph{weight
enumerator} of a linear code. Given a \emph{generating matrix}
$A\in\{0,1\}^{r\times C}$ of rank $r$, a \emph{code word}~$c$ is any
vector in the linear subspace~$\Upsilon$ generated by the rows of~$A$
over GF(2). For any real number~$\lambda$, the \emph{weight
enumerator} of the code is given by $W_A(\lambda)=\sum_{c\in
\Upsilon}\lambda^{\|c\|}$, where $\|c\|$ is the number of $1$'s
in~$c$. The problem of computing the weight enumerator of
a linear code is in \FP for $\lambda\in\{-1,0,1\}$,
and is known to be \nP-hard for every other fixed
$\lambda\in\Rationals$ (see \cite{welsh}).
We could not find a proof, so we
provide one here.  We restrict attention to positive~$\lambda$,
since that is adequate for our purposes.
\begin{lemma}
Computing the Weight Enumerator of a Linear Code is \nP-hard
for any fixed positive rational number $\lambda\neq 1$.
\label{WE}
\end{lemma}
\begin{proof}
We will prove hardness by reduction from a problem~$\text{\eval}(H)$,
for some appropriate~$H$, using
Corollary~\ref{bulgro}. Let the input to $\text{\eval}(H)$ be a connected
graph $G=(V,E)$ with $V=\{v_1,\ldots,v_n\}$ and
$E=\{e_1,\ldots,e_m\}$. Let $B$ be the $n\times m$ incidence matrix
of~$G$, with $b_{ij}=1$ if $v_i\in e_j$ and $b_{ij}=0$ otherwise.
Let $A$ be the $(n-1)\times m$ matrix which is $B$ with the row for
$v_n$ deleted. $A$ will be the generating matrix of the Weight
Enumerator instance, with $r=n-1$ and $C=m$. It has rank $(n-1)$ since $G$ contains a
spanning tree. A code word $c$ has $c_j=\bigoplus_{i\in U} b_{ij}$,
where $U\subseteq V\setminus\{v_{n}\}$. Thus $c_j=1$ if and
only if $e_j$ has exactly one endpoint in~$U$, and the weight of~$c$
is $\lambda^{k}$, where $k$ is the number of edges in the cut
$U,V\setminus U$. Thus $W_A(\lambda) = \frac12 Z_H(G)$, where $H$ is
the symmetric weight matrix with $H_{11}=H_{22}=1$ and
$H_{12}=H_{21}=\lambda$. The $\frac12$ arises because we fixed which
side of the cut contains $v_n$. Now $H$ has rank 2 unless
$\lambda=1$, so this problem is \nP-hard by Corollary~\ref{bulgro}.
Note, by the way, that $Z_H(G)$ is the partition function of the
Ising model in statistical physics~\cite{Cip87}.
\end{proof}

\section{The Proof of Lemma~\ref{seven}}

Throughout this section, we assume $q=2$.
The following Lemma is
a generalisation of a result of
Creignou and Hermann~\cite{CH}, which deals with the case
in which $f$ is a relation
(or, in our setting, a function with range $\{0,1\}$).
The inductive technique used in the proof of Lemma~\ref{lem:CH}
(combined with the follow-up in Lemma~\ref{four})
is
good for showing that $\nCSP(\mathcal{F})$ is \nP-hard when $\mathcal{F}$
contains a \emph{single} function.
A very different situation arises when $\nCSP(\{f\})$ and $\nCSP(\{g\})$
are in \FP but $\nCSP(\{f,g\})$ is \nP-hard due to \emph{interactions} between~$f$
and~$g$ --- we deal with that problem later.

\begin{lemma}
\label{lem:CH}
Suppose that $f\in \mathcal{F}_2$  does not have affine support.
Then $\nCSP(\{f\})$ is \nP-hard.
\end{lemma}

\begin{proof}

Let $k$ be the arity of~$f$,
and let us denote the $i^{\textrm{th}}$ component of $k$-tuple $a\in R_f$ by $a_i$.
The proof is by induction on~$k$.
The lemma is trivially true for
$k=1$, since all functions of arity~$1$ have affine support.

For $k=2$, we note that since $R_f$ is not affine, it is
of the form
$R_f = \{(\alpha,\beta),(\bar{\alpha},\beta),(\bar{\alpha},\bar{\beta})\}$
for some $\alpha\in\{0,1\}$ and $\beta\in\{0,1\}$.
We can show that $\nCSP(\{f\})$ is \nP-hard by reduction from \eval($H$)
using \[H=
\left(
 \begin{array}{cc}
f(0,0) & f(0,1) \\
f(1,0) & f(1,1)\\
\end{array}
\right),
\]
which has rank~$2$ and exactly one entry that is~$0$. Given an
instance $G=(V,E)$ of \eval($H$) we construct an instance $I$ of
$\nCSP(\{f\})$ as follows. The variables of~$I$ are the vertices
of~$G$. For each edge $e=(u,v)$ of~$G$, add a constraint with
function~$f$ and variable sequence $u,v$. Corollary~\ref{bulgro} now
tells us that \eval($H$) is \nP-hard, so $\nCSP(\{f\})$ is \nP-hard.

Suppose $k>2$.
We start with some general arguments and notation.
For any $i\in\{1,\ldots,k\}$ and any $\alpha\in\{0,1\}$
let $f^{i=\alpha}$ be the
function of arity~$k-1$ derived from~$f$ by
pinning the $i$'th position to~$\alpha$.
That is, $f^{i=\alpha}(x_1,\ldots,x_{k-1})=
f(x_1,\ldots,x_{i-1},\alpha,x_{i+1},\ldots,x_k)$.
Also, let $f^{i=*}$ be the
projection of~$f$ onto all positions apart from position~$i$
(see Section~\ref{sec:project}).
Note that $\nCSP(\{f^{i=\alpha}\})\redT \nCSP(\{f,\delta_0,\delta_1\})$,
since  $f^{i=\alpha}$ can obviously be simulated by
$\set{f, \delta_0,\delta_1}$.
Furthermore,
by Lemma~\ref{lem:2pinning}, $\nCSP(\{f,\delta_0,\delta_1\})\redT
\nCSP(\{f\})$.
Thus, we can assume that $f^{i=\alpha}$ has affine support
--- otherwise, we are finished
by induction.
Similarly, by Lemma~\ref{lem:project},
$\nCSP(\set{f^{i=*}})\redT\nCSP(\set{f})$. Thus we can assume that $f^{i=*}$ has affine support
---
otherwise, we are finished by induction.

Now, recall that $R_f$ is not affine. Consider any $a,b,c\in R_f$ such that
$d=a\oplus b \oplus c\notin R_f$.
We have 4~cases.

{\bf Case 1: There are indices $1\leq i<j\leq k$ such that
$(a_i,b_i,c_i)=(a_j,b_j,c_j)$:
$\>$}
Without loss of generality, suppose $i=1$ and $j=2$.
Define the function~$f'$ of arity~$(k-1)$ by
$f'(r_2,\ldots,r_k) = f(r_2,r_2,\ldots,r_k)$.
Note that $R_{f'}$ is not affine since
the condition $a\oplus b\oplus c\notin R_{f}$ is
inherited by $R_{f'}$.
So, by induction, $\nCSP(\{f'\})$ is \nP-hard.
Now note that
$\nCSP(\{f'\})\redT \nCSP(\{f\})$. To see this, note that any instance~$I_1$
of $\nCSP(\{f'\})$ can be turned into an instance~$I$ of
$\nCSP(\{f\})$ by repeating the first variable in the sequence of variables
for each constraint.

{\bf Case 2: There is an index $1\leq i \leq k$ such that
$a_i=b_i=c_i$:$\>$}
Since $d$ is not in $R_f$ and $d_i=a_i$,
we find that $f^{i=a_i}$
does not have affine support,
contrary to earlier assumptions.

Having finished Cases~1 and~2,
we may assume without loss
of generality that we are in Case~3 or Case~4 below, where
$\set{\alpha,\beta}\in\set{0,1}$,
$\bar{\alpha}=1-\alpha$,
$\bar{\beta}=1-\beta$ and $a',b',c'\in\set{0,1}^{k-2}$.

{\bf Case 3: $a=(\bar{\alpha},\bar{\beta},a')$,
$b=(\bar{\alpha},\beta,b')$, $c=(\alpha,\bar{\beta},c')$:$\>$}
Since $R_{f^{1=*}}$ is affine and $a$, $b$ and $c$ are in~$R_f$,
we must have either
$d=(\alpha,\beta,d')\in R_f$ or $e=(\bar{\alpha},\beta,d')\in R_f$, where
$d'=a'\oplus b'\oplus c'$. In the first case, we are done
(we have contradicted the assumption that $d\not\in R_f$), so assume
that $e\in R_f$ but $d\not\in R_f$.
Similarly, since $R_{f^{2=*}}$ is affine, we may assume that
$g=(\alpha,\bar{\beta},d')\in R_f$.
Since $R_{f^{1=\bar{\alpha}}}$ is affine
and $a$, $b$ and $e$ are in $R_f$,
we find that
$h=a\oplus b\oplus e=(\bar{\alpha},\bar{\beta},c')\in R_f$.
Since $R_{f^{2=\bar{\beta}}}$ is affine
and $a$, $c$ and $g$ are in $R_f$,
we find that
$i=(\bar{\alpha},\bar{\beta},b')\in R_f$.
Also, since
$R_{f^{2=\bar{\beta}}}$ is affine
and $a$, $h$ and $i$ are in $R_f$,
we find that
$j=(\bar{\alpha},\bar{\beta},d')\in R_f$.
Let $f'(r_1,r_2) = f(r_1,r_2,d_3,\ldots,d_k)$.
Since $e$, $g$ and $j$ are in~$R_f$ but $d$ is not,
we have
$(\bar{\alpha},\beta),(\alpha,\bar{\beta}),(\bar{\alpha},\bar{\beta})\in
R_{f'}$, but $(\alpha,\beta)\notin R_{f'}$.
Thus, $f'$ does not have affine support
and $\nCSP(\{f'\})$ is \nP-hard by induction.
Also, $\nCSP(\{f'\})\redT\nCSP(\{f\})$ by Lemma~\ref{lem:2pinning}.

{\bf Case 4: $a=(\bar{\alpha},\alpha,a')$, $b=(\bar{\alpha},\alpha,b')$,
$c=(\alpha,\bar{\alpha},c')$:$\>$}
Since $R_{f^{1=*}}$ is affine and $a$, $b$ and $c$ are in $R_f$ but $d$ is not,
we have $e=(\bar{\alpha},\bar{\alpha},d')\in R_f$.
Similarly, since $R_{f^{2=*}}$ is affine and $a$, $b$ and $c$ are in $R_f$ but
$d$ is not, we have $g=(\alpha,\alpha,d')\in R_f$.
Now since $R_{f^{1=\bar{\alpha}}}$ is affine and $a$, $b$ and $e$ are in $R_f$,
we have $h=(\bar{\alpha},\bar{\alpha},c')\in R_f$.
Also, since $R_{f^{2={\alpha}}}$ is affine and $a$, $b$ and $g$ are in $R_f$,
we have $i=(\alpha,\alpha,c')\in R_f$.

Let $f'(r_1,r_2) = f(r_1,r_2,c_3,\ldots,c_k)$.
If $j=(\bar{\alpha},\alpha,c')\not\in R_f$
then $f'$ does not have affine support (since $c$, $h$ and $i$ are in $R_f$)
so we finish by induction as in Case~3.
Suppose $j\in R_f$.
Since $R_{f^{1=\bar{\alpha}}}$ is affine
and $a$, $b$ and $j$ are in $R_f$, we have
$\ell=(\bar{\alpha},\alpha,d')\in R_f$.
Let $f''(r_1,r_2) = f(r_1,r_2,d_3,\ldots,d_k)$.
Then $f''$ does not have affine support (since $e$, $g$ and $\ell$ are in $R_f$
but $d$ is not)
so we finish by induction as in Case~3.
\end{proof}

Lemma~\ref{lem:CH} showed that
$\nCSP(\{f\})$ is \nP-hard when $f$ does not have affine support.
The following lemma gives another (rather technical, but useful) condition
which implies that
$\nCSP(\{f\})$ is \nP-hard.
We start with some notation.
Let $f$ be an arity-$k$ function.
For a value $b\in\{0,1\}$, an index
$i\in\{1,\ldots,k\}$,
and a tuple
$y\in\{0,1\}^{k-1}$,
let $\substuple{y}{i}{b}$ denote the tuple $x\in\{0,1\}^k$
formed by setting $x_i=b$
and $x_j=y_j$ $(j\in\set{1,\ldots,k}\setminus\set{i})$.

We say that index~$i$ of $f$ is \emph{useful} if there is a
tuple~$y$ such that $f(\substuple{y}{i}{0})>0$ and
$f(\substuple{y}{i}{1})>0$. We say that $f$ is {\it product-like\/}
if, for every useful index~$i$, there is a rational
number~$\lambda_i$ such that, for all $y\in\{0,1\}^{k-1}$,
\begin{equation}\label{condition}
f(\substuple{y}{i}{0})=\lambda_i f(\substuple{y}{i}{1}).\end{equation}
If every position~$i$ of~$f$ is useful then being product-like is the same
as being of product type. However, being product-like is less demanding
because it does not restrict indices that are not useful.

\begin{lemma}
\label{hard}
\label{four}
If $f\in \mathcal{F}_2$
is not product-like then {\rm\#CSP}$(\{f\})$ is \nP-hard.
\end{lemma}

\begin{proof}
We'll use Corollary~\ref{bulgro} to prove hardness, following an
argument from~\cite{dgp}. Choose a useful index~$i$ so that there is
no $\lambda_i$ satisfying~(\ref{condition}).

Suppose $f$ has arity~$k$.
Let $A$ be the $2 \times 2^{k-1}$ matrix
such that for $b\in\{0,1\}$ and $y\in \{0,1\}^{k-1}$,
$A_{b,y}=f(\substuple{y}{i}{b})$. Let $A'=A A^T$.

First, we show that \eval($A'$) is \nP-hard.
Note that $A'$ is the following symmetric $2\times 2$ matrix
with non-negative rational entries.
\[
\left(
\begin{array}{cc}
\sum_y A_{0,y}^2 &\sum_y A_{0,y}A_{1,y}\\
\sum_y A_{0,y}A_{1,y} & \sum_y A_{1,y}^2\\
\end{array}
\right)
=
\left(
\begin{array}{cc}
\sum_y {f(\substuple{y}{i}{0})}^{2}
   &\sum_y f(\substuple{y}{i}{0})f(\substuple{y}{i}{1})\\
\sum_y f(\substuple{y}{i}{0})f(\substuple{y}{i}{1})
   & \sum_y f(\substuple{y}{i}{1})^2\\
\end{array}
\right)
\]
Since index~$i$ is useful, all four entries of~$A'$ are positive. To
show that \eval($A'$) is \nP-hard by Corollary~\ref{bulgro}, we just
need to show that its determinant is non-zero.  By Cauchy-Schwartz,
the determinant is non-negative, and is zero only if $\lambda_i$
exists, which have assumed not to be the case. Thus \eval($A'$) is
\nP-hard by Corollary~\ref{bulgro}.

Now we reduce \eval($A'$) to \#CSP$(\{f\})$.
To do this, take an undirected graph $G$ which is an instance
of \eval($A'$). Construct an instance~$Y$ of \#CSP$(\{f\})$.
For every vertex~$v$ of~$G$ we introduce a variable $x_v$ of $Y$.
Also, for every edge $e$ of~$G$ we introduce $k-1$ variables
$x_{e,1},\ldots,x_{e,{k-1}}$ of~$Y$.
We introduce constraints in~$Y$ as follows. For each edge
$e=(v,v')$ of~$G$ we introduce constraints
$f(x_v,x_{e,1},\ldots,x_{e,k-1})$ and
$f(x_{v'},x_{e,1},\ldots,x_{e,k-1})$ into~$Y$,
where we have assumed, without loss of generality, that
the first index is useful.

It is clear that \eval($A'$) is exactly equal to the partition function
of the \#CSP$(\{f\})$ instance~$Y$.
\end{proof}

For $w\in\Rationals^+$, let $U_{w}$ denote the unary function mapping 0 to~1 and 1 to~$w$.
Note that
$U_0=\delta_0$,
and $U_1$ gives the constant (0-ary function) 1,
occurrences of which leave the partition function unchanged.
So, by Lemma~\ref{lem:2pinning}, we can discard these constraints since they do not add to the complexity of the problem. Note, by the observation above about proportional functions, that
the functions  $U_w$ include all unary functions except for $\delta_1$ and
the constant 0.
We can discard $\delta_1$ by Lemma~\ref{lem:2pinning}, and if
the constant 0 function is in $\mathcal{F}$, any instance $I$ where it appears as a constraint has $Z(I)=0$. So again we can discard
these constraints since they
not add to the complexity of the problem.

Thus $U_w$ will be called \emph{nontrivial} if $w\notin\set{0,1}$.
Let  $\oplus_{k}:\{0,1\}^{k}\to\{0,1\}$ be the arity-$k$ parity
function that is 1 iff its argument has
an odd number of $1$s. Let
$\neg\oplus_{k}:\{0,1\}^{k}\to\{0,1\}$ be the function
$1-\oplus_{k}$.
The following lemma shows that even a simple function like $\oplus_3$
can lead to intractable \#CSP instances when it is combined with
a nontrivial weight function $U_\lambda$.

\begin{lemma}
$\nCSP(\oplus_{3},U_{\lambda},\delta_0,\delta_1)$ and
$\nCSP(\neg\oplus_{3},U_{\lambda},\delta_0,\delta_1)$
are both \nP-hard, for any positive $\lambda\neq 1$.
\label{one}
\end{lemma}
\begin{proof}
We give a
reduction from
computing the
Weight Enumerator of a Linear Code, which was shown to be
\nP-hard in Lemma~\ref{WE}.
In what follows, it is sometimes convenient to
view $\oplus_{k}$, $\delta_{0}$, etc., as relations
as well as functions to $\set{0,1}$.

We first argue that for any $k$, the relation
$\oplus_k$ can be simulated by $\set{\oplus_3, \delta_0, \delta_1}$.
For example, to simulate
$x_1 \oplus \cdots \oplus x_k$ for $k>3$, take new variables $y$, $z$ and $w$ and
let $m=\lceil k/2 \rceil$ and use $x_1 \oplus \cdots \oplus x_m \oplus y$
and $x_{m+1} \oplus \cdots \oplus x_k \oplus z$ and $y \oplus z\oplus w$
and $\delta_0(w)$.

Since $\set{\oplus_3, \delta_0, \delta_1}$ can be used to simulate
any relation $\oplus_k$, we can use
$\set{\oplus_3, \delta_0, \delta_1}$
to simulate an arbitrary system of linear equations over
$\mathrm{GF}(2)$.
In particular we can use them to simulate the subspace
$\Upsilon$
of code words for a given generating matrix $A$.

Finally, we can use $U_{\lambda}$ to simulate the
function which evaluates the weight enumerator on $\Upsilon$.
Then, since $\lambda\neq 0,1$,  we can  apply Lemma~\ref{WE} to complete the argument.
The same proof, with minor modifications, applies to $\neg\oplus_3$.
\end{proof}

\begin{lemma}
Suppose
$f\in\mathcal{F}_2$ is not of
product type.  Then,
for any
positive $\lambda\neq 1$, there exists a constant~$c$,
depending on $f$, such that
$\nCSP(\{f,\delta_0,\delta_1,U_{\lambda},U_{c}\})$ is \nP-hard.
\label{two}
\end{lemma}
\begin{proof}

If $f$ does not have affine support, the result follows by
Lemma~\ref{lem:CH}. So suppose $f$ has affine support.
Consider the underlying relation $R_f$, viewed as a table.
The rows of the table represent the tuples of the relation.
Let $J$ be the set of columns on which the relation is not constant.
That is, if $i\in J$ then there is a row~$x$ with $x_i=0$ and a row~$y$ with $y_i=1$.
Group the columns in~$J$
into equivalence classes:  two columns are equivalent iff
they are equal or complementary.  Let $k$ be the number of
equivalence classes.
Take one column from each
of the $k$~equivalence classes as a representative, and focus
on the arity-$k$ relation~$R$ induced by those columns.\vspace{1ex}

\noindent{\bf Case 1: Suppose $R$ is the complete relation of arity~$k$.}\\
Let $f^*$ be the projection of~$f$ onto the $k$ columns of~$R$.
By Lemma~\ref{lem:project},\vspace{-1.4ex}
\[\nCSP(\{f^*\})
\, \redT\, \nCSP(\{f\})\, \redT\,
\nCSP(\{f,\delta_0,\delta_1,U_{\lambda},U_{c}\}).\vspace{-1.4ex}\]
We will argue that
$\nCSP(\{f^*\})$ is \nP-hard. To see this, note that every column
of~$f^*$ is useful. Thus, if $f^*$ were product-like, we could
conclude that $f^*$ was of product type.  But this would imply that
$f$ is of product type, which is not the case by assumption. So
$f^*$ is not product-like and hardness follows from
Lemma~\ref{hard}.\vspace{1ex}

\noindent{\bf Case 2: Suppose $R$ is not the complete relation of
arity~$k$.}\\
We had assumed that $R_f$ is
affine. This means that given three vectors, $x$, $y$ and $z$, in $R_f$,
$x\xor y \xor z$ is in $R_f$ as well.
The arity-$k$ relation~$R$ inherits this property,
so is also affine.

Choose a minimal set of columns of~$R$ that
do not induce the complete relation.  This exists by assumption.
Suppose there are $j$~columns in this minimal set.  Observe that
$j\neq 1$ because there are no constant columns in~$J$.   Also $j\not=2$,
since otherwise the two columns would be related by equality
or disequality, contradicting the preprocessing step. The argument here is that on two
columns, $R$ cannot have exactly three tuples because it is affine,
and having tuples $x$, $y$ and $z$ in would require
the fourth tuple $x\xor y \xor z$.
But if it has two tuples then,  because there are no constant columns, the
only possibilities are either $(0,0)$ and $(1,1)$, or $(0,1)$ and $(1,0)$.
Both contradict the preprocessing step, so $j\geq3$.

Let $R'$ be the restriction of~$R$ to the $j$~columns.
Now $R'$ of course has fewer than $2^j$ rows, and at least $2^{j-1}$
by minimality.  It is affine, and hence must be
$\oplus_j$ or~$\neg\oplus_j$. To see this, first note that the
size of $R'$ has to be a power of~$2$ since $R'$ is the solution to a system
of linear equations. Hence the size of $R'$ must be $2^{j-1}$. Then,
since there are $j$ variables, there can only be one defining equation.
And, since every subset of $j-1$ variables induces a complete relation,
this single equation must involve all variables. Therefore, the equation is $\xor_j$ or
$\neg\xor_j$.

Let $f'$ be the projection of~$f$ onto the $j$ columns just identified.
Let $f''$ be further
obtained by pinning all but three of the $j$~variables to~0.
Pinning $j-3$ variables to~$0$ leaves a single
equation involving all three remaining variables.
Thus  $R_{f''}$ must be $\xor_3$ or $\neg \xor_3$.

Now define the symmetric function $f'''$ by\vspace{-1.4ex}
\[f'''(a,b,c) =
f''(a,b,c)\*f''(a,c,b)\*f''(b,a,c)\*f''(b,c,a)\*f''(c,a,b)\*f''(c,b,a),\vspace{-1.4ex}\]
Note that $R_{f'''}$ is $\oplus_3$ or $\neg\oplus_3$, since $R_{f''}$ is symmetric
and hence $R_{f'''}=R_{f''}$.

To summarise: using~$f$
and the constant functions~$\delta_0$ and~$\delta_1$, we have simulated
a function $f'''$ such that its underlying relation
$R_{f'''}$ is either $\oplus_3$ or $\neg\oplus_3$.
Furthermore, if triples~$x$ and~$y$ have the
same number of $1$s then $f'''(x)=f'''(y)$.

We can now simulate an unweighted version of $\oplus_3$ or
$\neg\oplus_3$ using $f'''$ and a unary function $U_c$, with $c$
set to a conveniently-chosen value. There
are two cases. Suppose first that the affine support of $f'''$ is
$\neg\xor_3$. Then let $w_{0}$ denote the value of $f'''$ when
applied to the $3$-tuple $(0,0,0)$ and let $w_2$ denote
$f'''(0,1,1)=f'''(1,0,1)=f'''(1,1,0)$.  Recall that $f'''(x)=0$ for any
other $3$-tuple~$x$.  Now let $c={(w_0/w_2)}^{1/2}$.
Note from the definition of $f'''$ that $w_{0}$ and $w_{2}$ are
squares of rational numbers, so $c$ is also rational. Define a
function~$g$ of arity~3 by $g(\alpha,\beta,\gamma) =
U_c(\alpha)U_c(\beta) U_c(\gamma)f'''(\alpha, \beta, \gamma)$. Note
that $g(0,0,0)=w_0$ and $g(0,1,1)=g(1,0,1)=g(1,1,0) = c^2 w_2 = w_0$. Thus,
$g$ is a pure affine function with affine support $\neg\xor_3$ and
range $\{0,w_0\}$. The other case, in which the affine support of
$f'''$ is $\xor_3$, is similar.

We have established a reduction from either
$\nCSP(\oplus_{3},U_{\lambda},\delta_0,\delta_1)$
or $\nCSP(\neg\oplus_{3},U_{\lambda},\delta_0,\delta_1)$,
which are both \nP-hard by Lemma~\ref{one}.
\end{proof}

\begin{lemma}
\label{five} If $f\in\mathcal{F}_2$ is not of product type, then
$\nCSP(\{f,\delta_0,\delta_1,U_{\lambda}\})$ is \nP-hard for any positive~$\lambda\neq 1$.
\end{lemma}
\begin{proof}
Take an instance $I$ of $\nCSP(\{f,\delta_0,\delta_1,U_{\lambda},U_{c}\})$,
from Lemma~\ref{two}, with $n$ variables $x_1,x_2,\ldots,x_n$. We want
to compute the partition function~$Z(I)$ using only instances of
$\nCSP(\{f,\delta_0,\delta_1,U_{\lambda}\})$. That is, instances which avoid using
constraints~$U_c$.
For each $i$, let $m_{i}$ denote the number of copies
of $U_{c}$ that are applied to $x_i$, and let $m=\sum_{i=1}^n m_{i}$.
Then we can write the partition function
as $Z(I)=Z(I;c)$ where
\[Z(I;w) = \sum_{\sigma\in{\{0,1\}}^n}
\hat{Z}(\sigma) \prod_{i:\sigma_{i}=1}w^{m_{i}}=\sum_{\sigma\in{\{0,1\}}^n}
\hat{Z}(\sigma) w^{\sum_{i=1}^n m_{i}\sigma_i},\]
where $\hat{Z}(\sigma)$ denotes the value corresponding
to the assignment~$\sigma(x_i)=\sigma_i$, ignoring constraints applying $U_c$,
and $w$ is a variable.
So $\hat{Z}(\sigma)$ is the weight of~$\sigma$, taken over all constraints
other than those applying $U_{c}$. Note also that $Z(I;w)$ is a polynomial
of degree $m$ in $w$. We can evaluate $Z(I;w)$ at the point $w=\lambda^j$
by replacing each $U_c$ constraint with $j$ copies of a $U_\lambda$
constraint. This evaluation is an instance of
$\nCSP(\set{f,\delta_0,\delta_1,U_{\lambda}})$. So, using $m$ different values of~$j$
and interpolating, we learn the coefficients of the polynomial
$Z(I;w)$. Then we can put $w=c$ to evaluate $Z(I)$.
\end{proof}

\begin{lemma}
\label{six} Suppose $f\in\mathcal{F}_2$   is not of product type, and
$g\in\mathcal{F}_2$   is not pure affine.
Then $\nCSP(\{f,g,\delta_0,\delta_1\})$
is \nP-hard.
\end{lemma}
\begin{proof}
If $g$ does not have affine support we are done by Lemma~\ref{lem:CH}.
So suppose that $g$ has affine support.
Since $g$ is not pure affine,
the range of $g$ contains at least two non-zero
values.

The high-level idea will be to use pinning and bisection to extract a
non-trivial unary weight function~$U_\lambda$ from~$g$. Then we can reduce from
$\nCSP(\{f,\delta_0,\delta_1,U_{\lambda}\})$, which we proved \nP-hard in Lemma~\ref{five}.

Look at~the relation $R_g$, viewed as a table.
If every column were constant, then $g$ would be pure affine, so this
is not the case. Select a non-constant column~with index $h$.
If there are two non-zero values
in the range of $g$
amongst
the rows of~$R_g$ that are~$0$ in column~$h$
then we derive a new function $g'$ by pinning column~$h$
to~$0$. The new function $g'$ is not pure affine, since the two
non-zero values prevent this. So we will show inductively that
$\nCSP(\set{f,g',\delta_0,\delta_1})$ is \nP-hard.
This will give the result since
$\nCSP(\{f,g',\delta_0,\delta_1\})$ trivially reduces to
$\nCSP(\{f,g,\delta_0,\delta_1\})$.

If we don't finish this way, or
symmetrically by pinning column~$h$ to~$1$, then we know
that there are distinct positive values~$w_0$ and~$w_1$ such that, for every
row~$x$ of~$R_g$ with~$0$ in column~$h$, $g(x)=w_0$ and, for every
row~$x$ of~$R_g$ with~$1$ in column~$h$, $g(x)=w_1$.
Now note that, because the underlying relation $R_g$ is affine,
it has the same number of $0$'s in column~$h$ as $1$'s.
This is because $R_g$ is the solution of a set of linear equations.
Adding the equation $x_h=0$ or $x_h=1$ exactly halves the set of solutions
in either case. We now project onto the index set $\set{h}$. We obtain
the unary weight function $U_\lambda$, with $\lambda=w_1/w_0$, on
using the earlier observation about proportional functions.
This was our goal, and completes the proof.
\end{proof}
Lemma~\ref{seven} now follows from Lemma~\ref{lem:2pinning} and Lemma~\ref{six},
completing the proof of Theorem~\ref{thm:main}.

\section{Appendix}

The purpose of this appendix is to prove Lemma~\ref{lem:pinning} for
an arbitrary fixed domain $[q]$. We
used only the special case $q=2$, which we
stated and proved as Lemma~\ref{lem:2pinning}. However, pinning appears to be
a useful technique for studying the complexity of \nCSP, so we
give a proof of the general Lemma~\ref{lem:pinning}, which we
believe will be applicable elsewhere.
\begin{appendixlemma}
For every $\mathcal{F}\subseteq \mathcal{F}_q$,
$\nCSP(\mathcal{F}\cup \bigcup_{c\in[q]}\delta_c) \redT \nCSP(\mathcal{F})$.
\end{appendixlemma}
In order to prove the lemma, we introduce a useful, but less natural,
variant of \nCSP.
Suppose
$\mathcal{F}\subseteq \mathcal{F}_q$.
An instance $I$ of $\nCSP^{\neq}(\mathcal{F})$ consists of
a set $V$ of variables and a set $\mathcal{C}$ of constraints, just
like an instance of $\nCSP(\mathcal{F})$.
In addition, the instance may contain a \emph{single} extra constraint~$C$
applying the arity-$q$ \emph{disequality} relation $\chi_{\neq}$
with scope $(v_{C,1},\ldots,v_{C,q})$.

The disequality relation $\chi_{\neq}$ is defined by
$\chi_{\neq}(x_1,\ldots,x_q)=1$ if $x_1,\ldots,x_q\in[q]$ are pairwise distinct.
That is, if they are a permutation of the domain $[q]$.
Otherwise, $\chi_{\neq}(x_1,\ldots,x_q)=0$.

Lemma~\ref{lem:pinning} follows immediately from Lemma~\ref{lem:first} and~\ref{lem:second} below.

\begin{lemma}
\label{lem:first}
For every
$\mathcal{F}\subseteq \mathcal{F}_q$,
$\nCSP(\mathcal{F}\cup \bigcup_{c\in[q]}\delta_c) \redT\nCSP^{\neq}(\mathcal{F})$.
\end{lemma}
\begin{proof}

We follow the proof lines of Lemma~\ref{lem:2pinning}, but instead of
subtracting the contribution corresponding to configurations in which some $t_i$'s get
the same value, we use the disequality relation to restrict the partition function
to configurations in which they get distinct values.

Say that $\mathcal{F}$ is \emph{symmetric} if it is the case that
for every arity-$k$ function $f\in\mathcal{F}$ and every
tuple $x\in[q]^k$ and every permutation $\pi:[q]\rightarrow[q]$,
$f(x_1,\ldots,x_k) = f(\pi(x_1),\ldots,\pi(x_k))$.

Let
$I$ be an instance of $\nCSP(\mathcal{F}\cup \bigcup_{c\in[q]}\delta_c)$
with variable set $V$. Let $V_c$ be the set of variables $v\in V$ to which
the constraint $\delta_c(v)$ is applied. Assume without loss of generality
that the sets $V_c$ are pairwise disjoint.
Let $V_{q}=V \setminus\bigcup_{c\in[q]} V_c$.
We construct an instance $I'$ of $\nCSP^{\neq}(\mathcal{F})$.
The instance has variables $V_q \cup \{t_0,\ldots,t_{q-1}\}$.
Every constraint~$C$ of~$I$
involving a function $f\in\mathcal{F}$
corresponds to a constraint~$C'$ of~$I'$. Here
$C'$ is the same as~$C$ except that variables in~$V_c$ are replaced with~$t_c$,
for each $c\in[q]$. Also, we add a new disequality constraint
to the new variables $t_0,\ldots,t_{q-1}$.

\noindent{\bf Case 1.
$\mathcal{F}$ is symmetric:\quad}

By construction, $Z(I')= \sum_{y_0,\ldots,y_{q-1}} Z(I'\mid
\sigma(t_0)=y_0,\ldots,\sigma(t_{q-1})=y_{q-1})$, where the sum is over all permutations
$y_0,\ldots,y_{q-1}$ of~$[q]$.
By symmetry, the summands are all the same, so
$Z(I') = q! Z(I' \mid \sigma(t_0)=0,\ldots,\sigma(t_{q-1})=q-1) =
q! Z(I)$.

\noindent{\bf Case 2. $\mathcal{F}$ is not symmetric:\quad}

Say that two permutations $\pi_1:[q]\rightarrow[q]$ and $\pi_2:[q]\rightarrow[q]$
are \emph{equivalent} if, for every $f\in\mathcal{F}$ and every
tuple $x\in[q]^k$,  $f(\pi_1(x_1),\ldots,
\pi_1(x_k)) = f(\pi_2(x_1),\ldots,\pi_2(x_k))$.
Partition the permutations $\pi:[q]\rightarrow[q]$ into equivalence classes.
Let $h$ be the number of equivalence classes and $n_i$ be the size of the $i$'th
equivalence class, so $n_1 + \cdots + n_h=q!$.\footnote{In fact, it can be shown that
these equivalence classes are cosets of the symmetry group of~$f$,
and hence are of equal size, though we do not use this fact here.}
Let $\{\pi_1,\ldots,\pi_h\}$ be a set of representatives
of the equivalence classes with $\pi_1$ being the identity.
We know that $n_1\neq q!$ since $\mathcal{F}$ is not symmetric.

For a positive integer~$\ell$
we will now build an instance $I'_\ell$ by adding new constraints to~$I'$.
For each $\pi_i$ other than $\pi_1$ we add constraints as follows.
Choose a function~$f_i\in\mathcal{F}$
and a tuple~$y$
such that $f_i(y_1,\ldots,y_k)\neq f_i(\pi_i(y_1),\ldots,\pi_i(y_k))$.
If $f_i(y_1,\ldots,y_k)> f_i(\pi_i(y_1),\ldots,\pi_i(y_k))$ then
define the $k$-tuple $x^i$ by $(x^i_1,\ldots,x^i_k)
=(y_1,\ldots,y_k)$.
Otherwise, let $n$ be the order of the permutation $\pi_i$
and let $g_r$ denote $f_i(\pi_i^r(y_1),\ldots,\pi_i^r(y_k))$.
Since $g_0<g_1$ and $g_n=g_0$ there exists a $\xi\in\{1,\ldots,n-1\}$
such that $g_\xi>g_{\xi+1}$.
Let $(x^i_1,\ldots,x^i_k)=(\pi^{\xi}(y_1),\ldots,\pi^{\xi}(y_k))$ so
$f_i(x^i_1,\ldots,x^i_k)> f_i(\pi_i(x^i_1),\ldots,\pi_i(x^i_k))$.

Let $w_{ij}$ denote $f_i(\pi_j(x^i_1),\ldots,\pi_j(x^i_k))$ so,
since $\pi_1$ is the identity, we have just ensured that $w_{i1}>w_{ii}$.
Let $s^i=(t_{x_1^i},\ldots,t_{x_k^i})$, and let $0\leq z_i\leq h$ $(i=2,\ldots,h)$ be
positive integers, which we will determine below.
Add $\ell z_i$  new constraints to $I'_\ell$ with relation~$f_i$ and scope~$s^i$.
Let $\lambda_i = \prod_{\gamma=2}^h w_{\gamma i}^{z_\gamma}$.
Note that, given $\sigma(t_0)=\pi_i(0),\ldots,\sigma(t_{q-1})=\pi_i(q-1)$,
the contribution to $Z(I'_\ell)$ for the new constraints is
\[
\prod_{\gamma=2}^h f_\gamma(\sigma(t_{x^\gamma_1}),\ldots,\sigma(t_{x^\gamma_k}))^{z_\gamma \ell}=
\prod_{\gamma=2}^h f_\gamma(\pi_i(x^\gamma_1),\ldots,\pi_i(x^\gamma_k))^{z_\gamma \ell}=
\prod_{\gamma=2}^h w_{\gamma,i}^{z_\gamma \ell}
= \bigg(\prod_{\gamma=2}^h w_{\gamma,i}^{z_{\gamma}} \bigg)^\ell=
{\lambda_i}^\ell.\]
So
\[Z(I'_\ell) = \sum_{i=1}^h n_i\,
Z(\,I'\mid \sigma(t_0)=\pi_i(0),\ldots,\sigma(t_{q-1})=\pi_i(q-1)\,)
\,\lambda_i^\ell.\]

We have ensured that $\lambda_1>0$,
since $w_{i1}>w_{ii}\geq 0$, so $w_{i1}>0$ for all~$i=2,\ldots,h$.
We now choose the $z_i$'s so that $\lambda_i\neq \lambda_1$ for all $i=2,\ldots,h$.
If $w_{\gamma i}=0$ for any $\gamma=2,\ldots,h$, we have $\lambda_i=0$ and hence
$\lambda_i\neq\lambda_1$. Thus we will assume, without loss of generality, that $w_{\gamma i}>0$ for all $\gamma=2,\ldots,h$ and $i=2,\ldots,h'$, where $h'\leq h$.
Then we have
\[\frac{\lambda_i}{\lambda_1}\ =\ \prod_{\gamma=2}^h \Big(\frac{w_{\gamma i}}{w_{\gamma 1}}\Big)^{z_\gamma}\ =\ e^{\sum_{\gamma=2}^h\alpha_{\gamma i}z_\gamma}
\qquad (i=2,\ldots,h'),\]
where $\alpha_{\gamma i}=\ln(w_{\gamma i}/w_{\gamma 1})$. Note that $\alpha_{ii}< 0$,
since $w_{ii}<w_{i1}$. We need to find an integer vector $z=(z_2,\ldots,z_h)$
so that none of the linear forms $\mathcal{L}_i(z)=\sum_{\gamma=2}^h\alpha_{\gamma i}z_\gamma$ is zero, for $i=2,\ldots,h'$. We do this using a proof method similar
to the Schwartz-Zippel Lemma. (See, for example,~\cite{Schwar80}.) None of the $\mathcal{L}_i(z)$ is identically zero, since  $\alpha_{ii}\neq 0$. Consider the integer
vectors $z\in[h]^{h-1}$. At most $h^{h-2}$ of these can make $\mathcal{L}_i(z)$ zero for
any $i$, since the equation $\mathcal{L}_i(z)=0$ makes $z_i$ a linear function of $z_\gamma$
($\gamma\neq i$). Therefore there are at most $(h'-1)h^{h-2}<h^{h-1}$ such $z$ which make any
$\mathcal{L}_i(z)$ zero. Therefore there must be a vector $z\in[h]^{h-1}$ for which none of
the $\mathcal{L}_i(z)$ is zero, and this is the vector we require.

Now, by combining terms with equal $\lambda_i$
and ignoring terms with $\lambda_i=0$,
we can view $Z(I'_\ell)$
as a sum $Z(I'_\ell) = \sum_i c_i \lambda_i^{\ell}$
where
the $\lambda_i$'s are positive and pairwise distinct
and
\[c_1 = n_1 Z(I'\mid \sigma(t_0)=0,\ldots,\sigma(t_{q-1})=q-1).\]
Thus, by Lemma~3.2 of~\cite{DG} we can interpolate to recover $c_1$.
Dividing by $n_1$, we get
\[Z(I'\mid \sigma(t_0)=0,\ldots,\sigma(t_{q-1})=q-1)=Z(I).\qedhere\]
\end{proof}

\begin{lemma}
For every $\mathcal{F}\subseteq \mathcal{F}_q$,
$\nCSP^{\neq}(\mathcal{F}) \redT \nCSP(\mathcal{F})$.
\label{lem:second}
\end{lemma}
\begin{proof}

We use M\"obius inversion for posets, following the lines of the proof of~\cite[Theorem~8]{BD}.\footnote{Lov\'asz~\cite{Lovasz67} had previously used
M\"obius inversion in a similar context.}
Consider the set of partitions of $[q]$. Let $\underline{0}$
denote the partition with $q$ singleton classes. Consider
the partial order in which
$\eta\leq \theta$ iff every class of $\eta$ is a subset of some class of~$\theta$.
Define $\mu(\underline{0})=1$ and
for any $\theta\neq \underline{0}$
define $\mu(\theta)=-\sum_{\eta\leq \theta,\eta\neq\theta} \mu(\eta)$.
Consider the sum $\sum_{\eta \leq \theta} \mu(\eta)$. Clearly, this sum is~$1$
if~$\theta=\underline{0}$. From the definition of $\mu$,
it is also easy to see that the sum is~$0$ otherwise, since
\[\sum_{\eta \leq \theta} \mu(\eta) = \mu(\theta)+
\sum_{\eta \leq \theta,\eta\neq \theta} \mu(\eta)
=0.\]

Now let $I$ be an instance of $ \nCSP^{\neq}(\mathcal{F})$
with a disequality constraint applied to variables $t_0,\ldots,t_{q-1}$.
Let $V$ be the set of variables of~$I$.
Given a configuration $\sigma:V \rightarrow[q]$, let
$\vartheta(\sigma)$ be the partition of~$[q]$ induced by
of $(\sigma(t_0),\ldots,\sigma(t_{q-1}))$. Thus $i$ and $j$ in $[q]$ are in
the same class of $\vartheta(\sigma)$ iff $\sigma(t_i)=\sigma(t_j)$.
We say that a partition $\eta$ is consistent with $\sigma$
(written $\eta \preccurlyeq \sigma$) if $\eta\leq \vartheta(\sigma)$.
Note that $\eta \preccurlyeq \sigma$ means that for any $i$ and $j$
in the same class of $\eta$, $\sigma(t_i)=\sigma(t_j)$.

Let $\Omega$ be the set of configurations $\sigma$ that
satisfy all constraints in~$I$ except possibly the disequality constraint.
Then
$Z(I)=\sum_{\sigma\in\Omega} w(\sigma) \mathds{1}_\sigma$, where $\mathds{1}_\sigma=1$
if $\sigma$ respects the disequality constraint, meaning that $\vartheta(\sigma)=
\underline{0}$, and $\mathds{1}_\sigma=0$~otherwise.
By the M\"obius inversion formula derived above,
\[Z(I)\,=\,\sum_{\sigma\in\Omega} w(\sigma)
\sum_{\eta \leq \vartheta(\sigma)} \mu(\eta).\]
Changing the order of summation, we get
\[Z(I)\,=\,\sum_{\eta} \mu(\eta) \sum_{\eta\leq \theta}\hspace{1pt}
\sum_{\sigma\in \Omega:\vartheta(\sigma)=\theta} w(\sigma)\,=\,
\sum_{\eta} \mu(\eta) \sum_{\sigma\in \Omega: \eta \preccurlyeq \sigma}
w(\sigma).\]

Now note that $\sum_{\sigma: \eta \preccurlyeq \sigma} w(\sigma)$ is
the partition function $Z(I_\eta)$ of an instance $I_\eta$ of
$\nCSP(\mathcal{F})$. The instance $I_\eta$ is formed from~$I$ by
ignoring the disequality constraint, and identifying variables in
$t_0,\ldots,t_{q-1}$ whose indices are in the same class of~$\eta$.
Thus we can compute all the $Z(I_\eta)$ in $\nCSP(\mathcal{F})$.
Finally, $Z(I) = \sum_{\eta} \mu(\eta) Z(I_\eta)$, completing the
reduction.
\end{proof}

\end{document}